\documentclass[11pt,a4paper]{article}
\usepackage{a4wide}
\usepackage{amsmath}
\usepackage{amssymb}
\usepackage{amsthm}
\usepackage{graphicx}
\usepackage{color}
\usepackage{float}
\usepackage{nicefrac}
\usepackage{authblk}
\usepackage{cite}
\usepackage{color}
\usepackage{xcolor}

\theoremstyle{plain} \newtheorem{claim}{Claim}
\theoremstyle{plain} \newtheorem{lemma}{Lemma}
\theoremstyle{plain} \newtheorem{theorem}{Theorem}
\theoremstyle{plain} \newtheorem{definition}{Definition}
\theoremstyle{plain} 
\theoremstyle{plain} 
\theoremstyle{remark} 
\theoremstyle{remark}

\floatstyle{ruled}
\newfloat{algorithm}{htbp}{loa}
\floatname{algorithm}{Algorithm}

\title{\textsc{\textbf{Integrality Gaps for Colorful Matchings}}}

\date{}

\author[1]{Steven Kelk \thanks{steven.kelk@maastrichtuniversity.nl}}
\author[1]{Georgios Stamoulis \thanks{georgios.stamoulis@maastrichtuniversity.nl}}
\affil[1]{Department of Data Science and Knowledge Engineering (DKE), Maastricht University, The Netherlands}

\begin{document}

\maketitle

\begin{abstract}
We study the integrality gap of the natural linear programming relaxation for the \textit{Bounded Color Matching} (BCM) problem. We provide several families of instances and establish lower bounds on their integrality gaps and we study how the Sherali-Adams ``lift-and-project" technique behaves on these instances. We complement these results by showing that if we exclude certain simple sub-structures from our input graphs, then the integrality gap of the natural linear formulation strictly improves. To prove this, we adapt for our purposes the results of F\"{u}redi [\emph{Combinatorica}, 1(2):155-162, 1981]. We further leverage this to show upper bounds on the performance of the Sherali-Adams hierarchy when applied to the natural LP relaxation of the BCM problem.
\end{abstract}

%%%%%%%%%%%%%%%%%%%%%%%%%%%%%%%%%%%%%%%%%%%%%%%%%%%%%%%%%%%%%%%%%%%%%%%%%%%%%%%%%%%%%%%%%%%%%%%%%%%%%%%%%%%%%%%%%%%%%%%%%%%%%%%%%%%%%%%%%%%%5
%%%%%%%%%%%%%%%%%%%%%%%%%%%%%%%%%%%%%%%%%%%%%%%%%%%%%%%%%%%%%%%%%%%%%%%%%%%%%%%%%%%%%%%%%%%%%%%%%%%%%%%%%%%%%%%%%%%%%%%%%%%%%%%%%%%%%%%%%%%%5
\section{Introduction And Problem Definition}
%%%%%%%%%%%%%%%%%%%%%%%%%%%%%%%%%%%%%%%%%%%%%%%%%%%%%%%%%%%%%%%%%%%%%%%%%%%%%%%%%%%%%%%%%%%%%%%%%%%%%%%%%%%%%%%%%%%%%%%%%%%%%%%%%%%%%%%%%%%%5
%%%%%%%%%%%%%%%%%%%%%%%%%%%%%%%%%%%%%%%%%%%%%%%%%%%%%%%%%%%%%%%%%%%%%%%%%%%%%%%%%%%%%%%%%%%%%%%%%%%%%%%%%%%%%%%%%%%%%%%%%%%%%%%%%%%%%%%%%%%%5
In 1982, Papadimitriou \& Yannakakis defined the \textit{Exact Matching} (EM) problem \cite{DBLP:journals/jacm/PapadimitriouY82}: Given a bipartite graph $B$ with some edges painted red, does $B$ contain a perfect matching with \textit{exactly} $k \in \mathbb{Z}^+$ red edges? This is one of the very few problems whose complexity is not yet fully understood. On one hand, there exists an exact polynomial time randomized \textbf{NC} algorithm by Mulmuley and U. \& V. Vazirani  \cite{DBLP:journals/combinatorica/MulmuleyVV87} which suggests that EM is probably not \textbf{NP}-complete.  Moreover, Yuster \cite{DBLP:journals/algorithmica/Yuster12} showed that there exists an algorithm which, in polynomial time, either correctly decides that there is no maximum matching with exactly $k$ red edges or returns a matching of cardinality at most $\mu(G)-1$ with exactly $k$ red edges,  where $\mu(G)$ is the matching number of the input graph $G$ i.e., the maximum cardinality matching in $G$. This  result puts EM as close to \textbf{P} as possible (unless of course EM $\in \mathbf{P}$). The problem was also studied in some restricted classes, for example in complete and complete bipartite graphs, see Karzanov and Yi, Murty \& Spera  \cite{Karzanov1987,Yi2002261} respectively. Still, the exact complexity of the problem remains unknown and this has prompted researchers to investigate meaningful related cases of the Exact Matching problem.

Here we consider the following very natural generalization of the EM problem:

\begin{definition}[Bounded Color Matching-BCM]
We are given a (simple, undirected) graph $G=(V,E)$. The edge set $E$ is partitioned into $k$ sets $E_1 \cup  \cdots \cup E_k$ i.e., every edge $e$ belongs to color class $C_j$ if $e \in E_j$ and has a profit $p_e$ $\in \mathbb{Q}^+$. By abusing notation slightly, we will say that edge $e$ ``has" color $C_j$ if $e \in E_j$. Let $\mathcal{C} = \cup_{i=1, \dots, k} C_i$ be the collection of all color classes. Each color class $C_j$ is associated with a positive number $w_j \geq 1$. Our goal is to find a maximum (weighted) matching $M$ that contains at most $w_j$ edges of color $C_j$ i.e., a matching $M$ such that $|M \cap E_j|  \leq w_j, ~\forall C_j \in \mathcal{C}$.
\end{definition}

In  \cite{DBLP:conf/mfcs/Stamoulis14} an LP-based approximation algorithm with approximation ratio 1/2 was given for the BCM problem, which matches the \emph{integrality gap} of the natural LP relaxation for this problem. The algorithm is based on the elegant technique  by Parekh \cite{DBLP:conf/ipco/Parekh11} which gives an inductive process to write any basic feasible solution of the relaxed LP as an approximate \textit{sparse} convex combination of integral solutions.  The result holds for any bounds $w_j \geq 1$, integral or otherwise, since the analysis does not make use of the fact that $w_i \in \mathbb{Z}^+, \forall i$, only the fact that $w_i \geq 1$ (otherwise the integrality gap could be unbounded). It has been further generalized by Parekh and Pritchard \cite{DBLP:conf/waoa/ParekhP14} to uniform hypergraphs.

A very natural question occurs: a negative result based on a bad integrality gap instance rules out the possibility of a good relaxation-based approximation algorithm. But this holds only for the \emph{particular} relaxation that we use. What about other, more complicated and sophisticated relaxations? As an illustrative example, if we take the normal (degree-constrained) relaxation for the classical matching problem, which has integrality gap of $\nicefrac{3}{2}$, and enhance it with the \textit{blossom inequalities}, we get an exact formulation of the convex hull of all integer points for the matching problem \cite{Edmonds1965b}.

Given the apparent difficulty of identifying stronger/tighter linear relaxations for combinatorial optimization problems, a large body of work has been dedicated in recent years to identifying systematic techniques to enhance the quality of a given linear (or semi-definite) program with \textit{valid inequalities} (inequalities that are satisfied by all \textit{integral} points). The hope is that the part of the polyhedron responsible for the bad integrality gap example will be eliminated.  Many such ``lift and project" methods have been proposed so far, in particular by Sherali and Adams (SA) \cite{DBLP:journals/siamdm/SheraliA90}, by Lov\'asz and Schrijver (LS) \cite{Lovasz91conesof}, by Balas, Ceria and Cornu\'ejols  (BCC) \cite{DBLP:journals/mp/BalasCC93}, by  Lasserre \cite{DBLP:journals/siamjo/Lasserre02} and by Bienstock and Zuckerberg (BZ) \cite{DBLP:journals/siamjo/BienstockZ04}.  For a very thorough and readable comparison of the first three such hierarchies see \cite{DBLP:journals/mor/Laurent03}. Their use in approximation algorithms was initiated by the seminal work of Arora, Bollob\'{a}s, Lov\'{a}sz and Tourlakis \cite{DBLP:journals/toc/AroraBLT06}.

The general idea has the following pattern: Let $P_0 = \{ \mathbf{x} \in \{0,1\}^n:\mathbf{Ax} \leq \mathbf{b} \}$, $\mathbf{A} \in \mathbb{R}^{m \times n}, \mathbf{b}   \in \mathbb{R}^m$ be an initial integral polyhedron in $n$-th dimensional space and let $F_0$ be the corresponding relaxation i.e., $F_0 = \{\mathbf{x} \in  [0,1]^n: \mathbf{Ax} \leq \mathbf{b}\}$.  Starting from $F_0$ we operate in \emph{rounds} (also called \emph{levels}), and in each round new variables are added and a specific set of valid linear or semi-definite inequalities is added (the lifting phase) and then the lifted polyhedron is projected back to the original space (projection phase). Thus we obtain a \textit{hierarchy} of tighter formulations $F_k \subseteq F_{k-1} \subseteq \cdots \subseteq F_0$ of $F_0$ such that for each $0 < j \leq n$, $F_j$ is obtained from $F_{j-1}$. An important feature of this sequence is that we can efficiently optimize any linear (or semi-definite) objective function over $F_t$ for any fixed $t$ and, moreover, after at most $n$ rounds we have that $F_n = P_0 = conv(F_0 \cap \{0,1\}^n)$. That is, this progressively tighter sequence of relaxations converges to the convex hull of the integral solutions.

From the point of view of approximation algorithms, the first ``few" rounds of such hierarchies (constant or poly-logarithmic) are particularly interesting, especially for problems for which the gap between the current best approximation algorithm and the complexity theoretic inapproximability bound is large enough; the hope is that better (quasi-)polynomial algorithms can be designed. The effect of such methods has been extensively studied for a host of combinatorial optimization problems, for  example see \cite{DBLP:conf/stoc/MathieuS09,
DBLP:conf/ipco/KarlinMN11, DBLP:journals/mp/ChanL12, DBLP:journals/mp/CheriyanGGS16, DBLP:conf/wads/ChlamtacFG13,DBLP:conf/ipco/KurpiszMMMVW16,
DBLP:conf/approx/ChlamtacS08} and the references therein. In many cases such hierarchies fail to generate polytopes with better integrality gaps (after a few rounds) but there are some notable results where the current best approximation algorithms are known to be either consistent with
few rounds of some hierarchy or produce even better approximability results not achievable by other techniques. See, for example, \cite{DBLP:conf/stoc/AroraC06, DBLP:conf/approx/ChlamtacS08, DBLP:conf/ipco/KarlinMN11, DBLP:journals/mp/ChanL12, DBLP:conf/innovations/YoshidaZ14, DBLP:conf/icalp/ManurangsiR17, DBLP:conf/stoc/BateniCG09, DBLP:conf/stoc/BansalSS16} for some important works in that direction.

Of particular interest in our paper is the \textit{Sherali-Adams} (SA) Hierarchy, which we formally define in a  subsequent section. This is a very well-known and commonly-used
``lift-and-project" method in combinatorial optimization and has produced a host of positive results. See \cite{DBLP:conf/approx/MagenM09} for Vertex Cover in Planar graphs,
\cite{DBLP:conf/soda/VegaK07} for Max-Cut in dense graph instances, \cite{DBLP:conf/stoc/BateniCG09} for  Max-Min Fair Allocations,
\cite{DBLP:conf/innovations/YoshidaZ14, DBLP:conf/icalp/ManurangsiR17} for dense instances of CSPs and \cite{DBLP:conf/approx/ChlamtacKR10, DBLP:conf/stoc/GuptaTW13} for Sparsest Cuts in
bounded treewidth graphs. In contrast to these positive results, we will show that the Sherali-Adams hierarchy  is not so successful for the problems considered in this paper.

\smallskip
\noindent \textbf{\textsl{Our Contribution:}} We study to what extent formulations generated by the Sherali-Adams hierarchy can improve the integrality gap of the natural LP formulation for BCM. As a first step, we show that if we allow the bounds $w_j$ on the color classes to be fractional numbers (greater than or equal to 1, otherwise the integrality gap can be unbounded), then there exists a family of instances for the BCM problem such that even a sub-exponential number of rounds of the Sherali-Adams hiearchy does not suffice to reduce the integrality gap of 2. Similar bounds and instances (uniform lengths/sizes, fractional bounds/capacities) have been used in the integrality gap study of the Knapsack problem \cite{DBLP:conf/ipco/KarlinMN11}. This demonstrates a severe limitation of a more general computational model, i.e., even large families of large linear programs cannot ``realize" such relatively simple structured instances.

Given that the previous result for the Sherali-Adams hierarchy uses instances that have \emph{fractional} bounds, and because this might seem somewhat artificial, we next explore whether these bounds are inherently necessary. That is,  we would like to answer the following question: \textit{are there instances with integral color bounds $w_j$ with integrality gap of 2 and how does the Sherali-Adams hierarchy behave on them?} In that direction, we provide two extra families of integrality gap instances. First, a family of instances with integrality gap of 2 which resist only a constant number of Sherali-Adams rounds, and another family with integrality gap $\nicefrac{k}{k-1}$ for integer parameter $k$ which, in contrast, is preserved for a sub-exponential number of Sherali-Adams rounds. In order to show strong integrality gap properties for the Sherali-Adams hierarchy it is required that the instances have certain special properties (such as large degree and large cardinalities of the color classes).

This motivates the second part of our paper: there, we show that if we exclude a certain simple sub-structure (called \emph{truncated projective plane of order two}, i.e., an alternating bi-chromatic cycle on four vertices) then the integrality gap immediately improves. This means that every instance with integrality gap of 2 should have many disjoint copies of these sub-graphs. These bi-chromatic cycles, which cause the large integrality gap, can be recognized very quickly by the Sherali-Adams hierarchy i.e., only few rounds of this hierarchy are enough to lower the integrality gap below 2. For completeness we include the simple proof of this fact.

The combined results of our paper demonstrate that (i) the only instances of integrality gap 2 that resist a large (non-constant) number of rounds of the Sherali-Adams hierarchy, are instances with fractional bounds on the color classes and (ii) when we deal only with integral bounds, 2 rounds of the Sherali-Adams hierarchy suffice to reduce the integrality gap of the natural LP relaxation of the BCM problem. Our proofs are non-algorithmic and, although they were inspired by the results of F\"{u}redi \cite{DBLP:journals/combinatorica/Furedi81} (which were also used  by Chan \& Lau \cite{DBLP:journals/mp/ChanL12}), the technicalities involved make the arguments highly non-trivial. It remains a very interesting open problem to exploit this result algorithmically and this is a point that we will elaborate later in our manuscript.

\smallskip \noindent \textbf{\textsl{Related Work:}} To the best of our knowledge, the first time such a generalization of the EM problem was studied, at least from an approximation point of view, was in \cite{DBLP:conf/mfcs/NomikosPZ07} where the so-called \textit{blue-red} matching problem was studied: find a maximum cardinality matching with at most $w \in \mathbb{Z}^+$ red and at most $w$ blue edges. Besides the theoretical relevance, their  motivation was that  this can be used to approximately solve the \emph{Directed Maximum Routing and Wavelength Assignment} problem (DirMRWA) \cite{DBLP:conf/infocom/NomikosPZ03} in \textit{rings} which is a fundamental network topology \cite{DBLP:conf/mfcs/NomikosPZ07,DBLP:journals/siamdm/Caragiannis09}. They provided an $\mathbf{RNC^2}$ algorithm and a $\nicefrac{3}{4}$-approximation combinatorial algorithm noticing also that the greedy procedure produces a $\nicefrac{1}{2}$-approximate solution. The exact complexity of this problem remains open.

The BCM problem has appeared in the literature under many different names. In \cite{garey_johnson} ([GT55]) it was defined as \textit{Multiple Choice Matching} and was claimed to be \textbf{NP}-hard citing \cite{DBLP:journals/jacm/ItaiRT78}. Unfortunately, the results of \cite{DBLP:journals/jacm/ItaiRT78} do not prove this claim since the color classes do not form a partition of the edge set. This was acknowledged in \cite{DBLP:journals/dam/Rusu08} where it was shown, amongst other interesting results, that the problem is indeed \textbf{NP}-hard even on 3-regular bipartite graphs. BCM is also known as the \textit{Rainbow Matching} problem \cite{DBLP:journals/dam/Zaker07, DBLP:journals/jct/Woolbright78} when $w_j = 1, \forall j$.  In \cite{DBLP:journals/tcs/LeP14} a host of complexity results are given. Among these, it is shown that Rainbow Matching is hard to approximate within a factor better than $\nicefrac{139}{140}$ even in complete graphs and this trivially carries over to the BCM problem. Some graph classes where it is solvable in polynomial time were also identified.

Finally, the BCM problem can be recast as a problem of maximizing a linear function subject to a matching constraint and a partition matroid constraint which enforces that at most $w_j$ elements can be chosen from $C_j$. As a consequence, the greedy algorithm immediately gives a $\nicefrac{1}{3}$-approximation and this is tight i.e., there are simple instances where the greedy achieves exactly this ratio, see \cite{DBLP:conf/esa/Mestre06, tcs_matching_2012}.

\section{Technical Preliminaries}

Here we will define the natural linear programming formulation of the problem and we will comment on its properties with respect to its integrality gap. The purpose of the subsequent sections is to provide families of integrality gap instances and a study of the behavior of the Sherali-Adams hierarchy on them. We will give the standard definition of the Sherali-Adams hierarchy.

\medskip
For any vertex $v$ of a graph $G$ with edge set $E(G)$ let $\delta(v)= \{e \in E(G): v \in e\}$ i.e., the set of the edges incident to $v$. For a given instance of the BCM problem we can describe the set of all feasible solutions as follows.
\begin{eqnarray}
\label{bcm_polyhedron}
\mathcal{M}_c = \Bigg\{ \mathbf{x} \in \{0,1\}^{E} :~ \mathbf{x} \in \mathcal{M} \mbox{ and } \sum_{e \in E_j} x_e \leq w_j, ~\forall j  \in [k]    \Bigg\}
\end{eqnarray}

\noindent
where $\mathcal{M}$ is the usual (degree-constrained) matching polytope: $\mathcal{M} = \{ \mathbf{y} \in \{0,1\}^{E}: \sum_{e\in \delta(v)} y_e \leq 1, \forall v \in V(G) \}$. We call the additional constraints color constraints. We want to find the maximum profit solution vector $\mathbf{x}$ (that maximizes $\mathbf{p}^T \mathbf{x}$) such that $\mathbf{x} \in \mathcal{M}_c$. As usual, we relax the integrality constraints $\textbf{x} \in \{ 0,1\}^{E}$ to $\mathbf{x} \in [0,1]^{E}$ and we solve the corresponding linear relaxation efficiently to obtain a \textit{fractional}  vector $\mathbf{x} \in [0,1]^E$. It is not hard to show that the \textit{integrality gap} of $\mathcal{M}_c$ is 2 and this is true even if we add the blossom inequalities i.e., if instead of $\mathcal{M}$ as defined here, we use the well known Edmond's LP \cite{Edmonds1965b}.

Given an integral polyhedron $I$ for a maximization problem and its linear relaxation $L$ the integrality gap of $L$ is the maximum ratio of the optimal fractional solution over the optimal integral one, ranging over all possible instances. Linear relaxations that always have integral optimal solutions have integrality gap equal to 1. An LP formulation with integrality gap of $\varrho$ implies that it is impossible to design an approximation algorithm with performance guarantee better than $\varrho$ using this particular formulation as upper/lower bounding schema for our discrete optimization problem.

\smallskip
\noindent{\textbf{\textsl{The Sherali Adams Hierarchy:}}}
We recall the definition of the SA hierarchy of progressively stronger  relaxations of an integer polyhedron  in the $n$-dimensional hypercube $\{0,1\}^n$. We use the original
definition \cite{DBLP:journals/siamdm/SheraliA90}.

Let $F_0 = \{\mathbf{x} \in [0,1]^n : \mathbf{a_i}^T \mathbf{x} = \sum_{j \in [n]} a_{ij} x_j \leq b_i,  \forall i \in [m]  \}$ with $a_{ij}, b_i \in \mathbb{Q}$, $\forall i \in
[m], j \in [n]$ be an initial convex polyhedron in $[0,1]^n$ . Let $\mathcal{I} = \mathsf{conv} (F_0 \cap \{  0,1 \}^n)$ be the convex hull of all integer points of $F_0$. The
SA hierarchy, starting from $F_0$, constructs a hierarchy of progressively \textit{non-weaker} relaxations $F_1, F_2, \dots$ of $\mathcal{I}$ in the sense that $F_n
\subseteq F_{n-1} \subseteq \cdots \subseteq F_0$. Let $F_{\psi}$ be the polyhedron resulting after $\psi$ iterations of the SA methods applied initially to $F_0$. After at
most $n$ rounds we will arrive at $\mathcal{I}$ i.e., $F_n = \mathcal{I}$. Sometimes $n$ rounds are necessary in order to arrive at $\mathcal{I}$.  At the $\psi$-th iteration,
$\psi \geq 1$, the SA hierarchy obtains $F_{\psi}$ from $F_{\psi-1}$ (in fact, from $F_0$) as follows: For all disjoint subsets $\Gamma,\Delta$ of $[n]$ such that
$|\Gamma|+|\Delta| \leq \psi$:

\begin{description}
\item[\textsf{SA}-1] For each constraint $b_i - \mathbf{a_i}^T\mathbf{x}\geq 0$  add the constraint $$(b_i - \mathbf{a_{i}}^T\mathbf{x}) \prod_{\gamma \in \Gamma}x_{\gamma} \prod_{\delta \in
    \Delta}(1-x_{\delta}) \geq 0.$$
\item[\textsf{SA}-2]  Add all the constraints  $\prod_{\gamma \in \Gamma}x_{\gamma} \prod_{\delta \in \Delta}(1-x_{\delta}) \geq 0$.
\item[\textsf{SA}-3] Expand all the polynomial constraints described by SA-1 and SA-2:
\begin{itemize}
\item[(1)] Replace each term of the form $x_i^2$ by $y_{\{i\}}$,
\item[(2)]  Replace each product of monomials $\prod_{\zeta \in Z} x_{\zeta}$, defined by a set of variable indices $Z \subseteq \{1, \dots, n \}$, by a new variable $y_{Z}$.
\end{itemize}
\item[\textsf{SA}-4] Let $F^l_{\psi}$ be the resulting lifted polyhedron. Project $F^l_{\psi}$ onto the original $n$-th dimensional space by eliminating all $y_Z$ variables for which $|Z| \geq 2$:
$$F_{\psi} = \{\mathbf{x} \in [0,1]^{n}: \exists \mathbf{y} \in F_{\psi}^l, y_{\{i\}} = x_i,~ \forall i\}.$$
\end{description}

In case both $\Gamma, \Delta$ are empty, the corresponding term is simplified to $y_{\emptyset}$. The size of the lifted program after $t$ rounds is $\mathcal{O}(\sum_{i=1}^t \binom{n}{i})$. We note that the effect of the SA hierarchy on the usual matching polytope was fully studied in \cite{DBLP:conf/stoc/MathieuS09}. See also \cite{DBLP:conf/ipco/AuT11, DBLP:journals/mor/StephenT99, DBLP:journals/dam/AguileraBN04} for other relevant results regarding the performance of various lift-and-project methods on the matching polytope.

\section{Integrality Gaps for the Sherali-Adams Hierarchy}
We will show that the integrality gap of $\mathcal{M}_c$ resists an asymptotically linear number of rounds of the (SA) hierarchy by providing a particular family of graphs
and a feasible solution for the $\psi$-th level of the SA hierarchy with high fractional value with respect to the optimal integral value. We first provide our
integrality gap example.

\subsection{A family $\mathcal{F}$ of integrality gap instances}

For a given graph $G=(V,E)$, an \textit{edge coloring} of $G$ is a function $c: E \rightarrow \{1, \dots, k\}$ such that $c(e_1) \neq c(e_2)$ whenever $e_1 \cap e_2 \neq \emptyset$ (i.e., share a common endpoint). The \textit{edge chromatic number} (also known as \emph{chromatic index}) of a graph $G$ is the smallest positive integer  $k$ for which an edge coloring exists, and it is denoted by $\chi'(G)$. For any $G$, let $\Delta(G) = \max_{v \in V(G)} |\delta(v)|$. In a classical result, Vizing \cite{vizing} showed that $\Delta(G) \leq \chi'(G) \leq \Delta(G)+1$. For bipartite graphs a stronger statement holds:

\begin{theorem}[\cite{Konig1916}]
If $G$ is bipartite, then $\chi'(G) = \Delta(G)$.
\end{theorem}

Our starting point will be the $\ell$-dimensional hypercube graph $Q_{\ell}$: $Q_{\ell}$ can be constructed inductively from the disjoint union of the two hypercubes $Q_{\ell-1}$, by adding an edge from each vertex in one copy of $Q_{\ell-1}$ to the corresponding vertex in the other copy. The joining edges form a perfect matching. $Q_{\ell}$ has $2^{\ell}$ vertices and $\ell \cdot 2^{\ell-1}$ edges. More importantly, every hypercube graph is a uniform  bipartite graph of degree $\ell$ and thus, applying the result of \cite{Konig1916} we conclude that the edge chromatic number $\chi'$ of $Q_{\ell}$ is precisely $\ell$, the degree of each vertex in $V(Q_{\ell})$. In other words, we can edge-color the edges of $Q_{\ell}$ with $\ell$ colors such that all edges adjacent to any vertex receive distinct colors. Trivially, each color class $C_j, j \in  [\ell]$ contains $2^{\ell-1}$ edges. We set the bound $w_j = 2(1-\epsilon)$ for each color class, for some $\epsilon > 0$. Then, the maximum \textit{integral} matching contains $\ell$ edges (one edge per color class) whereas, by setting the values of the variables (corresponding to edges) to $\frac{1}{2^{\ell-2}} - \epsilon$, the maximum \textit{fractional} matching (solution to the LP relaxation $\mathcal{M}_c$) has value $\approx \frac{\ell \cdot 2^{\ell-1}}{2^{\ell-2}} = 2\ell$. This is indeed a feasible solution since (1) $\ell (\frac{1}{2^{\ell-2}}- \epsilon) \leq 1$ and (2) $2^{\ell-1} (\frac{1}{2^{\ell-2}} - \epsilon) \leq 2(1-\epsilon)$ i.e., it satisfies both degree and the color bound constraints.

\medskip
Let $\mathcal{F}$ be the family of all graphs constructed as
above. Observe that these  particular instances are ``easy" from an algorithmic point of view: indeed,  the first \emph{Chv\'{a}tal Closure} of $\mathcal{M}_c$ closes the gap. We
remind that the first Chv\'{a}tal closure of a polyhedron $P = \{ \mathbf{x} \in \mathbb{R}^n: \mathbf{Ax} \leq \mathbf{b} \}$ for $\mathbf{A} \in \mathbb{Q}^{m \times n}$ and $\mathbf{b} \in \mathbb{Q}^m$, is defined as
$$
P^c = \Bigg\{ \mathbf{x} \in \mathbb{R}^n: \mathbf{Ax} \leq \mathbf{b}, \sum_{j=1}^n(\lfloor \mathbf{u}^T \mathbf{A_j}\rfloor)x_j \leq \lfloor \mathbf{u}^T \mathbf{b} \rfloor, \forall \mathbf{u} \in \mathbb{R}^m  \Bigg\}.
$$
I.e., if we apply the first Chv\'{a}tal closure to $\mathcal{M}_c$ with $\mathbf{u} = (1,1,\dots, 1)^T$ vector for a graph $G \in \mathcal{F}$, then the integrality gap vanishes. On the
other hand, this closure alone is not enough to close the integrality gap on any \textit{arbitrary} instance: take the size four cycle with alternating edges from $E_1, E_2$ and
set $\beta_j = 2(1-\epsilon)$. Then, the first Chv\'{a}tal closure will set $b_j' = 1$ which has integrality gap again $2$ whereas two rounds of the (SA) are enough to eliminate this gap.
This shows that the two operators are incomparable, at least with respect with $\mathcal{M}_c$.

\smallskip
\noindent{\textbf{\textsl{The effect of SA on the family $\mathcal{F}$:}}} Let $G \in \mathcal{F}$ be any graph constructed as in the previous subsection for some $\ell$. Given such a $G$, we will define an appropriate fractional solution vector $\mathbf{y}$ and we will prove that $\mathbf{y}$ is feasible for the $\psi$-th level of the Sherali-Adams hierachy, for
\emph{any} $\psi = o(2^{\ell})$.
%, i.e., asymptotically the term $\psi$ is vanishing with respect to $\ell$ but still is within a constant factor of $\ell$.
Then we will see that this proposed vector has fractional value twice as large as the optimal integral solution.

Now define the vector  $\mathbf{y} \in F_{\psi}$  in $[0,1]^{\eta}$, $\eta = \sum_{q \in [\psi]}\binom{n}{q}$ as follows:
\begin{equation*}
\mathbf{y} = \left\{
\begin{array}{lclr}
~y_{ \emptyset} & = & 1 &\\
~y_{\{ e \}} & = & \frac{1-\epsilon}{2^{\ell-2}+\psi(1-\epsilon)} (= \rho), & \forall e \in E(G) \\
~y_{I} & = & 0, & \forall I \subseteq [n], |I| \geq 2
\end{array} \right.
\end{equation*}

We would like to show that this proposed vector is valid (feasible) for the $\psi$-th level of the SA hierarchy. In order to prove that, we need to prove that it satisfies all the
constraints of the $\psi$-th level of the SA hierarchy applied to $\mathcal{M}_c$ for a graph $G \in \mathcal{F}$. Analyzing the construction of the constraints of the $\psi$-th level of SA as outlined in the previous section, we have the following sets of constraints:

\smallskip \noindent \textbf{Degree constraints:} These correspond to all the constraints
\begin{displaymath}
\Big( 1 - \sum_{e \in \delta(v)} y_e \Big)  \prod_{\gamma \in \Gamma} y_{\gamma} \prod_{\delta \in \Delta} (1-y_{\delta}) \equiv \Big( 1 - \sum_{e \in \delta(v)} y_e \Big)
\sum_{H \subseteq \Delta} (-1)^{|H|} y_{\Gamma \cup H} \geq 0
\end{displaymath}
where $\Gamma, \Delta \subseteq [n]$: $\Gamma \cap \Delta = \emptyset$ and $|\Gamma|, |\Delta| \leq \min\{n, \psi +1 \}$. This is still not a linear constraint. If we insist
to fully linearize them, then they will take the form
\begin{displaymath}
\sum_{H \subseteq \Delta} (-1)^{|H|} y_{\Gamma \cup H} - \sum_{e \in \delta(v)}\sum_{H \subseteq \Delta} (-1)^{|H|} y_{\Gamma \cup H \cup \{ e \}} \geq 0.
\end{displaymath}

In the above $H, \Gamma$ are set of \emph{indices} of variables. By abusing notation slightly we allow ourself to write $\Gamma \cup H \cup \{ e \}$ where for $e$ we mean the
index of its corresponding variable. This is true in all the following.  We also use $y_e$ instead of $y_{\{e\}}$ (since the coordinates of $\mathbf{y}$ are defined on sets rather than elements).

\smallskip
\noindent \textbf{Color constraints:} Similarly, for all the color constraints we add all the constraints of the form
$$w_j \cdot \Big(\sum_{H \subseteq \Delta} (-1)^{|H|} y_{\Gamma \cup H} \Big) - \sum_{e \in E_j} \sum_{H \subseteq \Delta} (-1)^{|H|}  y_{\Gamma \cup H \cup \{ e \}}  \geq 0.$$

\smallskip
\noindent \textbf{Non-negativity constraints:} These are the constraints $1 -y_e \geq 0$ and $y_e \geq 0$, $\forall e \in E$. Identically with the previous cases, these
    constraints will become, respectively,
\begin{displaymath}
\sum_{H \subseteq \Delta}(-1)^{|H|} y_{\Gamma \cup \Delta} -  \sum_{H \subseteq \Delta } (-1)^{|H|} y_{\Gamma \cup H \cup \{e\}} \geq 0 ~~~\mbox{and}~  \sum_{H \subseteq
\Delta} (-1)^{|H|} y_{\Gamma \cup H \cup \{ e \}} \geq 0.
\end{displaymath}

%\smallskip We will show that the  vector $y$ previously defined satisfies all the above constraints.

\begin{lemma}
The vector $\mathbf{y}$, as defined above, is feasible for the $\psi$-th level of the Sherali-Adams hierarchy applied on $\mathcal{M}_c$, for any $\psi = o(2^{\ell-2})$.
\end{lemma}

\begin{proof}
First of all, it is immediate from the definition that $\mathbf{y}$ satisfies all the initial constraints (the constraint matrix of $\mathcal{M}_c$) or, in other words, the zero-th
level of the SA hierarchy applied to $\mathcal{M}_c$. We will prove that it satisfies all the color constraints arising after $\psi$ rounds, for any $\psi$. The other
two set of constraints  can be shown to be satisfied by the vector $\mathbf{y}$ using identical, and in fact easier, arguments. At the end, by selecting any $\psi = o(2^{\ell-2})$, we will prove that the value of the fractional solution is twice the value of the optimal integral one.

So, we have to show that for the defined $\mathbf{y}$ we have that
$$\Xi = \underbrace{w_j \sum_{H \subseteq \Delta} (-1)^{|H|} y_{\Gamma \cup H}}_{\Sigma_1}  - \overbrace{\sum_{e \in E_j}\sum_{H \subseteq \Delta} (-1)^{|H|} y_{\Gamma \cup H \cup \{ e \}}}^{\Sigma_2}  \geq 0. $$

To prove our claim, we will distinguish between three major cases with respect to the cardinality of the set $\Gamma$:

\begin{description}
\item[Case 1. $|\Gamma| \geq 2$:]

In this case we have that $|\Gamma \cup H| \geq 2$, $\forall H \subseteq \Delta$ and so, be the definition of the solution vector $\mathbf{y}$ we have that
    $y_{\Gamma \cup H} = 0$. So, both $\Sigma_1,\Sigma_2$ become zero forcing the entire sum $\Xi$ to be zero and thus the constraint is trivially satisfied.

\item[Case 2. $|\Gamma| = 1$:]

In this case $\Gamma$ contains the index of some edge $e \in E(G)$ and again, by slightly abusing notation, we can write that $\Gamma = \{e\}$. In that case, there are two possibilities regarding the set $\Delta$ which we need to handle.

    We will first show that $\{e\}$ cannot belong in the set $\Delta$. Indeed, assume $\{e\}\in\Delta$. Then $\Gamma \cap \Delta$ is not equal to $\emptyset$. Using this we will show that the
    whole sum $\Xi$ is zero (and this is the reason why we impose the requirement that $\Gamma \cap \Delta$ should be $\emptyset$): For this, let $H \subseteq \Delta$ such that $\{e\} \notin H$ (the case $H = \{e\}$ is treated completely symmetrically). Then, the corresponding term in the sum becomes $(-1)^{|H|} y_{\Gamma \cup H}$. Consider now the term $H \cup \{e\}$. The corresponding term in the sum is now

    $$(-1)^{|H|+1} y_{\Gamma \cup H \cup \{ e\}} = (-1)^{|H|+1} y_{\Gamma  \cup H},$$

(since $\{e\} \in \Gamma$, we have that $y_{\Gamma \cup H \cup\{e\}} = y_{\Gamma \cup H}$), a term that has opposite sign than $(-1)^{|H|} y_{\Gamma \cup H}$. So, the two terms cancel each other, and the whole sum is zero. This shows that if $\{e\} \in \Delta$ then $\Xi$ is satisfied.

We will consider now the case where $\{e\} \notin \Delta$. This is equivalent to $\Delta = \emptyset$ since, otherwise, we would have $|\Gamma \cup \Delta|>1$ and so, be definition of $\mathbf{y}$, $y_{\Gamma \cup \Delta} = 0$. In that case, the sum $\Sigma_1$ is of the form $w_j \cdot (-1)^{0}y_{\Gamma \cup \emptyset} = y_{\{e\}} = w_j \cdot \rho \geq 0$ and the sum $\Sigma_2$ becomes simply $\rho$ because the only surviving term for the outermost summation (over all indexes of edges in $E_j$) is for the particular $\{e\} = \Gamma$ since the term $y_{e \cup e'} = 0$ for $e'\neq e$, and so we have that $ \Xi = w_j\rho - \rho > 0$ and so the constraint $\Xi$ is again satisfied.

\item[Case 3. $|\Gamma| = 0$:]

In this case, we will derive expressions for $\Sigma_1, \Sigma_2$ and compare them to prove the claim.

We start with $\Sigma_1$ and we see that in this case the only terms that survive are the term $y_{\emptyset} = 1$ with coefficient $(-1)^0 = 1$ and all the terms of the form $y_{\{h\}}$ for $h \in \Delta$ with coefficient $(-1)^1 = -1$. We have $|\Delta| \leq \psi$ many such terms so, at the end, we have that
\begin{displaymath}
\Sigma_1 = w_j(1 - |\Delta| \rho).
\end{displaymath}

For $\Sigma_2$ we proceed as follows: first we fix an $e \in E_j$. For this $e$, the surviving terms of the second sum are the ones corresponding to $H = \emptyset$ and $H =
\{e\}$, if $e \in \Delta$. For $H = \emptyset$ the corresponding term becomes $(-1)^0 y_{\{e\}} = \rho$ and we have one such term for each $e \in E_j$. For $H = \{e\} \in E_j
\cap \Delta$ the corresponding term becomes $(-1)^1 y_{\{e\} \cup \{e\}} = -1y_{\{e\}} = - \rho$ and we have $|E_j \cap \Delta|$ many such terms. For all $e' \in \Delta$ such
that $e' \notin E_j$, the corresponding terms become $y_{\{e'\}\cup \{e\}} = 0$ by definition. So, all in all,
$$
\Sigma_2  =  \rho|E_j|  - \rho |E_j \cap \Delta|.
$$

Since $\Xi = \Sigma_1 - \Sigma_2$, we want to prove that $\Xi \geq 0$ which is equivalent from the above derivations on $\Sigma_1$ and $\Sigma_2$ to
\begin{eqnarray*}
\Xi & = & \Sigma_1 - \Sigma_2 \\
& = & \Big(w_j(1 - |\Delta| \rho)\Big) - \Big(|E_j|\cdot \rho - |E_j \cap \Delta| \cdot \rho \Big) \geq 0 \\
& \Leftrightarrow & |E_j|\rho - |E_j\cap \Delta|\rho \leq w_j (1-|\Delta|\rho).
\end{eqnarray*}

Since $|E_j\cap \Delta|\rho \geq 0$, we will show that
$$|E_j|\rho \leq w_j (1-|\Delta|\rho) \Leftrightarrow |E_j|\rho + w_j|\Delta|\rho \leq w_j$$

\noindent
which trivially implies that $\Xi \geq 0$. Indeed, using the fact that $|\Delta| \leq \psi$ and $|E_j| = 2^{\ell-1}$, we have that
\begin{eqnarray*}
|E_j|\rho + w_j|\Delta|\rho  & \leq & \rho\Big( |E_j| + \psi \cdot w_j \Big) \\
& = & \frac{(1-\epsilon)}{2^{\ell-2} + \psi(1-\epsilon)} \cdot \Big( 2^{\ell-1} + \psi \cdot 2(1-\epsilon) \Big) \\
& = & \frac{(1-\epsilon)}{2^{\ell-2} + \psi(1-\epsilon)} \cdot 2 \cdot \Big(2^{\ell-2} + \psi \cdot (1-\epsilon)\Big) \\
& = & 2(1-\epsilon) =  w_j,
\end{eqnarray*}

\noindent
as required and this concludes the proof that $\Xi \geq 0$ when $|\Gamma| = 0$.
\end{description}

\smallskip
We have proven that the proposed vector $y$  satisfies all the color constraints arising after at most $\psi$ rounds of the Sherali-Adams hierarchy and observe that the analysis above is independent of the actual value of $\psi$. If we want to retain the integrality gap of 2 we will show that choosing any $\psi = o(2^{\ell-2})$ achieves this. The rest of the constraints (non-negativity, degree) can be proven to be satisfied by the proposed solution vector $y$, for the same bounds on the number of rounds $\psi$, in an identical manner. Here we briefly mention the details for the remaining cases. In order to show that the degree constrains are satisfied by the proposed vector $y$, we follow the calculations as above, and indeed the first two cases go through in exactly the same way. For the case $|\Gamma| = 0$ and for the constraint imposed by a vertex $v$, everything boils down to showing that $\rho(|\delta(v)|+\psi) =  \rho (\ell + \psi) \leq 1$, which is trivially true by the definition of $\rho$. Regarding the non-negativity constraints, we will first show that $\sum_{H \subseteq \Delta} (-1)^{|H|} y_{\Gamma \cup H \cup \{e\}} \geq 0$, which corresponds to the constraint $y_e \geq 0,$ for all $e \in E$. Again, we distinguish three cases regarding the cardinality of $\Gamma$.

Fix a term corresponding to a non-negativity constraint for an edge $e$. If $|\Gamma| > 1$ then the corresponding term is zero by definition. If $|\Gamma| = 1$ then the corresponding term is non-zero only when $\Gamma = e$. Again as before, in both cases where $e \in \Delta$ and $e \notin \Delta$, we see that the constraint is satisfied: if $e \in \Delta$ then the constraint is equal to zero otherwise is equal to $y_e > 0$. If $|\Gamma| = 0$ again we distinguish two cases regarding whether $e \in \Delta$ or not. In both cases, identical arguments as before show that the constraint should be $\geq 0$. Indeed, if $e \in \Delta$ then the term (constraint) becomes zero (the only surviving terms are the one for $H = \emptyset$ and $H = \{e\}$), otherwise it becomes $y_e > 0$ (the only surviving term is the one for $H= \emptyset$). The second non-negativity constraint $\sum_{H \subseteq \Delta}(-1)^{|H|} y_{\Gamma \cup \Delta} -  \sum_{H \subseteq \Delta } (-1)^{|H|} y_{\Gamma \cup H \cup \{e\}}$ is a special case of the degree constraint and so its non-negativity follows directly from the non-negativity of the latter.
\end{proof}

We now bound the value of the objective function for this $\mathbf{y}$: $\forall \epsilon > 0$ and $\psi = o(2^{\ell-2})$
$$ \textsf{value} (\mathbf{y})  =   \lim_{\epsilon \rightarrow 0} \Big(\ell 2^{\ell-1}
\frac{(1-\epsilon)}{2^{\ell-2}+\psi(1-\epsilon)} \Big) = \lim_{\epsilon \rightarrow 0} 2\ell \frac{2^{\ell-2} (1-\epsilon)}{2^{\ell-2} + \psi(1-\epsilon)} = 2 \ell. $$

\begin{theorem}
For any $\epsilon >0$, there exist graphs $G$ on $n$ vertices and $m$ edges such that for any $\psi = o(m)$ the integrality gap of the $\psi$-th level of the Sherali-Adams hierarchy applied to $\mathcal{M}_c$ for $G$, is at least $\frac{2}{1+\vartheta}$, $\vartheta = o(1)$.
\end{theorem}

\subsection{Integrality gap instances with integral bounds}

The results of the previous section used the fact that the color bounds were fractional numbers so a very natural question is whether we can find instances for the BCM problem with integer color bounds that cause the SA to perform poorly on them (whereby the integrality gap of 2 resists a large number of SA rounds). As the results of the previous section suggest, we need highly structured instances in order to ``fool" the SA hierarchy: both the degrees of the vertices and the cardinalities of the color classes are required to be $\Theta(n)$ in order to have strong integrality gaps for the SA hierarchy after $O(n)$ rounds. It is not clear at all if such instances exist and, if they do, how they can be constructed. In the next section we will show how the only instances that have an integrality gap of 2 that resist a  large number of Sherali-Adams rounds, must have fractional color bounds.

Observe that it is an easy task to come up with arbitrary instances that have integrality gap of 2 (for the natural initial relaxation): Consider the following family of bipartite instances, $\mathcal{B}$, for the BCM (in fact the Rainbow Matching) problem: take $k$ copies of the $C_4$ graph, $k \in \mathbb{N}$, where $C_4$ is the usual 4-cycle. Let the $i$-th copy of $C_4$, $C_4^i,~1\leq i \leq k$, have vertices $\alpha_i^1, \alpha_i^2,\alpha_i^3, \alpha_i^4$. Let $E_i^r = \{ \{ \alpha_i^1, \alpha_i^2 \}, \{\alpha_i^3, \alpha_i^4  \} \}$ and $E_i^b = \{ \{ \alpha_i^1, \alpha_i^4 \}, \{\alpha_i^2, \alpha_i^3 \} \} $. Now, connect the $i$-th copy of $C_4$, $C_4^i,~1\leq i \leq k-1$ with the $(i+1)$-th as follows: add the edge $\{ \alpha_i^2, \alpha_{i+1}^1 \}$ and assign to this edge a new color,  say $c_{w}$. Add the edge $\{ \alpha_i^3, \alpha_{i+1}^4 \}$ and color it again with $c_{w}$. Connect $C_4^k$ with $C_4^1$ in same way as before and assign to the two new edges color $c_{w}$. All color bounds are set to 1. See Figure \ref{example_bcc}.

\begin{figure}[h!]
\centering
\includegraphics[scale=0.85]{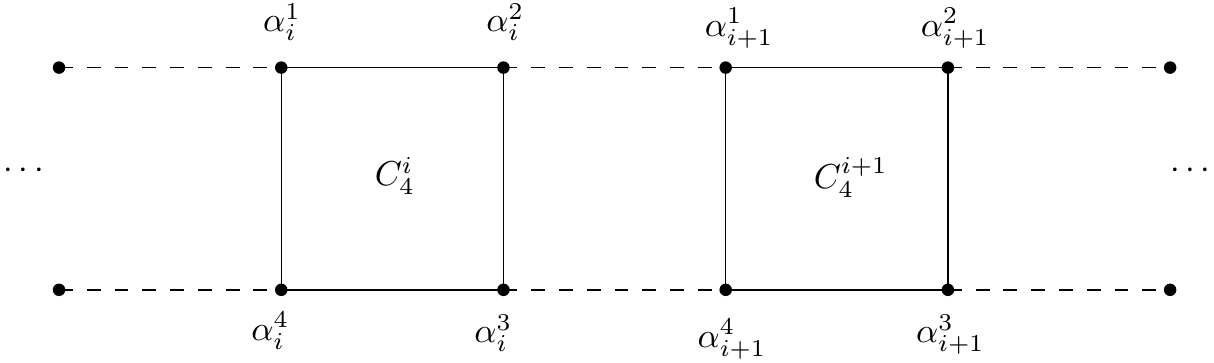}
\caption{A graph from $\mathcal{B}$.}
\label{example_bcc}
\end{figure}

All in all, our graph has $4k$ vertices and $4k + 2k = 6k$ edges, i.e.,  $E =  E_w \cup (\bigcup_{i=1}^k ( E_{i}^r \cup E_{i}^b) )$ where $|E_i^{j}| = 2$, $\forall i \in [k], j \in \{r,b \}$ and $|E_w| = 2k$.   The optimal integral solution has value $k+1$ whereas the optimal fractional solution can have value  $2k$. By following very similar calculations as in the previous section we can see that after only a few (constant, in fact 3) number of rounds the Sherali-Adams procedure will declare any fractional vector approaching value $2k$ as infeasible. This is consistent with the structural requirements explained above that are needed in order to prove large SA integrality gaps after a large number of rounds. As we will see in the next section, this is not a coincidence: these bi-chromatic cycles $C^i_4$ are, in some very precise sense, the only obstacles for instances with improved integrality gap bounds.

\smallskip
We move on by describing a third family $\mathcal{F}'$ of integrality gap instances. In contrast with the previous two families, this family will have integrality gap of $\nicefrac{k}{k-1}$ for parameter $k$. However, in contrast with the second family $\mathcal{B}$ described above, this bound on the integrality gap resists any sub-linear number of the Sherali-Adams strengthening. The construction is as follows: let $k = 2\ell$, where $\ell$ a positive integer greater than or equal than 1. The graph will be bipartite with bipartition $L,R$ where $|L| = |R| = k = 2\ell$. Moreover, the graph will be $k$-regular i.e., each vertex from the ``left" bipartition $L$ will be connected to each vertex of the ``right" bipartition $R$. We now provide the coloring of the edges to complete the instance of the BCM problem. The resulting graph will be properly edge-colored. We have $k = 2\ell$ different colors $c_0, c_2, \dots, c_{k-1}$. Take vertex $v_j \in L$, $j \in \{0,1, \dots, k-1\}$. For each $\Delta \in \{0, 1, \dots, k-1\}$ edge $(v_j, v_{(j + \Delta) \mod k})$ gets color $c_{\Delta}$.

\begin{figure}[h!]
\centering
\includegraphics[scale=0.90]{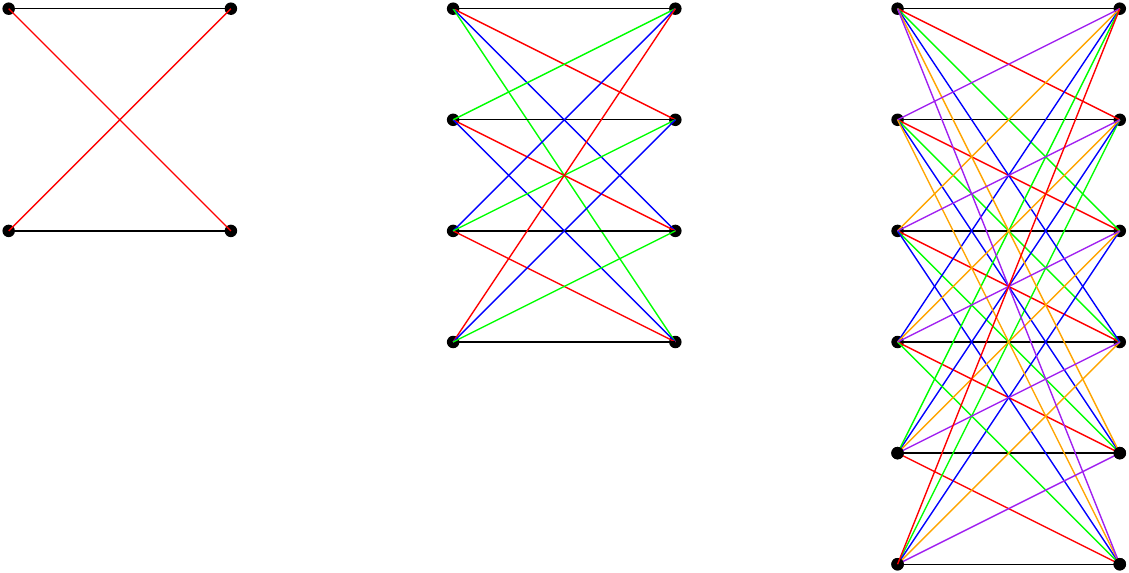}
\caption{The graphs corresponding to $\ell = 1,2,3$. The ordered list of colors is  black, red, blue, green, orange, purple and so, for example, black connects all vertices $v_i \in L$ to $u_i \in R$.}
\label{int_bounds}
\end{figure}

In other words, for every vertex $v_j \in L$ we take the ordered list of colors $c_0, \dots c_{k-1}$ and we paint $(v_j, u_j)$ with $c_0$, $(v_j,u_{j+1})$ with $c_1$ and so on. Set the bound of each color class equal to 1.  See Figure \ref{int_bounds} for a demonstration of the construction for $\ell = 1,2,3$.

For $\ell =1 (k=2)$ we have a bi-chromatic cycle  which has integrality gap of 2. For $\ell =2 (k=4)$ the integrality gap is 4 over 3: the maximum colored matching is 3 but the LP can have fractional value of 4. In general, for any even $k$, the maximum colored matching has cardinality of at most $k-1$ (see Lemma \ref{transversal_lemma} below)
%(every feasible colored matching of size $k-1$ blocks by construction the addition of an extra edge of the remaining color)
versus fractional value of $k$: set $x_e = \nicefrac{1}{k}$ for all edges $e$. Since the degree of every vertex and the cardinality of every color class are both $k$, this constitutes a feasible solution. Then, since we have $k \times k$ edges, the overall objective function value is $k$ giving an integrality gap of $\nicefrac{k}{k-1}$.

\begin{lemma}\label{transversal_lemma}
Let $G_k$ be a graph constructed as above for some even positive number $k$. Then, the cardinality of the maximum matching that has at most one edge per color is at most $k-1$.
\end{lemma}

\begin{proof}
Let $A$ be the matrix whose $(i,j)$ entry has the color of the edge $(v_i, u_j) \in L \times R$. It is easy to notice that $A$ is in fact a \textit{Latin Square} of order $k$ i.e., a matrix filled with $k$ different symbols (the colors), each color occurring exactly once in each row and exactly once in each column. Furthermore, by construction, $A$ corresponds to a \textit{cyclic group} of order $k$ (for every row, the list of colors is shifted by one with respect to the preceding row). If there was a feasible (colorful) matching of size $k$, then this would correspond to a \textit{transversal} in $A$: A transversal of a Latin Square is a set of entries which includes exactly one entry from each row and column and one of each symbol. But Latin Squares of cyclic groups of \textit{even} order cannot have a transversal \cite{transversals_latin_squares}.
\end{proof}

We note that the choice of $k$ being an even number is crucial: According to the well known Ryser's conjecture \cite{ryser} every Latin Square of order $n$ contains a Latin transversal when $n$ is odd. In our context this is equivalent to saying that every proper $n$-edge-coloring of the complete bipartite graph $K_{n,n}$ always contains a colorful perfect matching when $n$ is an odd number. Although this is a conjecture for general properly edge colored complete bipartite graphs, we can easily prove that this is true for the cyclic group corresponding to the coloring of the edges given above: Take the edge $(v_0, u_0)$ of color $c_0$. For $j=1, \dots, k-1$, vertex $v_j$ is matched to vertex $u_{(j+j) \mod k}$ and this edge $(v_j, u_{(j+j) \mod k})$ gets color $c_j$.  This constitutes a perfect colorful matching. Alternatively, the greedy strategy gives always such a matching. This implies that Ryser's conjecture is true for Latin Squares corresponding to cyclic groups.

Now, identical arguments as before (we omit the details since they are completely identical) give us that $o(k)$ rounds of the Sherali-Adams strengthening are not enough to reduce the integrality gap below $\nicefrac{k}{k-1}$ in the slightest.

\begin{theorem}
Let $G_k \in \mathcal{F}'$. Then, even  $o(k)$ rounds of the Sherali-Adams hierarchy applied to the natural LP relaxation of the BCM problem on this instance, are not enough to reduce the integrality gap below $\frac{k}{k-1}$.
\end{theorem}

\section{Improved Integrality Gap Bounds}
In this section we will study more carefully the integrality gap properties of the natural linear relaxation $\mathcal{M}_c$ of the BCM problem. The previous section suggests that the bi-chromatic cycles on four vertices are building blocks of instances of integrality gap 2. Here, we will formalize this result in the following strong sense: if we exclude these simple sub-structures (bi-chromatic cycles with alternating colors, like the $C_4^i$s above) from our input graphs, then the integrality gap \emph{strictly improves}. Towards that goal, we will firstly cast the problem as a natural \emph{hypergraph} matching problem. In order to provide an upper bound on the fractional value for a given instance (as a function of two relevant parameters: the size of its matching and the number of disjoint copies of these sub-structures) of the natural linear relaxation of BCM (as hypergraph matching problem), we will use the dual relaxation of $\mathcal{M}_c$: the value of any feasible solution to this dual program will provide an upper bound on the feasible fractional value of $\mathcal{M}_c$ (including the optimal value of it). We will then relate this value to the optimal integral solution. We will distinguish between the cases where the input graph instance is a bipartite graph or not, and give slightly different bounds for these two cases, although the idea is identical.

A direct implication of this is the following: if we want to construct instances of the BCM problem for which the Sherali-Adams hierarchy cannot close the integrality gap of 2 after a large (i.e., sub-linear) number of rounds, then fractional bounds are necessary. This is because we will show that the absence of the bi-chromatic cycles immediately reduces the integrality gap and, moreover, as the results of the previous section suggest, the Sherali-Adams hierarchy very quickly recognizes such instances (declares vectors which assign fractional value of 2 to each such cycle as infeasible).

\smallskip
Without any loss, we will focus on the case where $w_j = 1$ for all color classes $C_j \in \mathcal{C}$ i.e., the Rainbow Matching problem. We can easily cast this case as a hypergraph matching problem as follows: let $G = (V,E_1 \dots E_k)$ be an instance of this rainbow matching problem. For each color class $C_j$, create a new vertex $c_j$ and let $N$ be the set of all these new vertices. For every edge $e = \{u,v\} \in E_j$ of the initial graph, create the hyperedge $\{u,v,c_j\}$. In this way we have created a uniform (each edge has three elements) hypergraph $H = (V \cup N, E_H)$ where $E_H$ is the set of hyperedges constructed as above. It is immediate that any feasible matching in $H$ translates 1-1 to a feasible matching of $G$ with exactly the same cardinality. Given an instance where $w_j = k > 1$ for some $j$,  obtain the following 3-hypergraph matching problem  by introducing $k$ new color classes $C_{j_1}, \dots, C_{j_k}$ and set $w_{j_i} = 1, \forall i \in [k]$. For every edge $e = \{u,v\}$ of color $C_j$, include all the hyperedges $(u,v, C_{j_i})$. Any hypergraph matching of cardinality $c$ in the new hypergraph can be transferred in an immediate way to a feasible solution for the initial BCM instance of the same cardinality, i.e., to a solution that can have at most $k$ edges of color $C_j$.

Now, for every hyperedge $e \in E_H$, we introduce a binary variable $x_e$. Then the standard integer linear formulation of this hypergraph matching problem is simply to maximize $\sum_{e \in E_H} x_e$ subject to $\sum_{e: v \in e} x_e \leq 1$, for all $v \in V(H)$. By relaxing the integrality constraints to $x_e \in [0,1]$ for all hyperedges $e$ we obtain the linear relaxation of this LP which, as we have already discussed, has integrality gap of 2. Let us call this LP $\mathcal{HM}_c$.

Let us take the minimal instance that has integrality gap of 2 for the BCM (and Rainbow Matching) problem: a simple bi-chromatic 4-cycle with alternating edges of these two colors. It is easy to observe that if we cast this instance as a hypergraph instance, then this is equivalent to the \emph{truncated 3-uniform projective plane}. We remind that a projective plane is a hypergraph that satisfies the following conditions: (1) for any two vertices of the hypergraph, there is a \emph{unique} hyperedge that contains them both, (2) for any two hyperedges, they share exactly one common vertex, and (3) there are four vertices of the hypergraph such that no hyperedge contains more than two of them. It is a well known fact that $r$-uniform projective planes exist if $r-1$ is a \emph{prime power} (see \cite{Matousek:1998:IDM:552237}, page 250). A \emph{truncated} projective plane is obtained by removing a single vertex from the initial projective plane and all the hyperedges incident to that vertex.  Interestingly, (truncated) projective planes are linked to integrality gaps of the hypergraph matching problem (since in a projective plane we can choose exactly one independent hyperedge): an $r$-uniform projective plane has integrality gap of $r-1+\frac{1}{r}$  whereas a truncated $r$-uniform projective plane has integrality gap of $r-1$ for their corresponding natural LP relaxations. The 3-uniform projective plane is known as the \emph{Fano} plane. The projective plane we obtain by truncating it is simply the bi-chromatic 4-cycle with alternating edges from the two colors. We denote such sub-instances by $BC$ and by $BC_H$ we denote their hypergraph translation.

We move on by defining the dual LP of the one described by $\mathcal{HM}_c$: given a 3-uniform hypergraph $H$, i.e., an instance of the hypergraph representation of the Rainbow Matching problem, for every edge we have a constraint and for every vertex $v$ of $H$ a variable $y_v$. Then, for every hyperedge $e$ of $H$ we have the constraint $\sum_{v \in e} y_v \geq 1$. This is the dual of the hypergraph matching relaxation and any feasible fractional solution to it provides an upper bound on the fractional solution of the linear relaxation of the hypergraph matching problem.  By duality, the two optimal values are the same. Let $y^*$ denote the optimal (minimum) fractional dual value for a given instance. This dual LP, let us call it $D(\mathcal{HM}_c)$, is also called a fractional covering (or transversal) LP.

\begin{theorem}\label{main_thm}
Let $H$ be a 3-uniform hypergraph (a hypergraph instance for the rainbow matching problem) such that $H$ has a matching (independent set of edges) of size $\mu \in \mathbb{Z}^+$. Assume that $H$ has at most $q$ pairwise disjoint copies of $BC_H$. Then, we have that

\begin{enumerate}
\item $y^*(H)  \leq  \nicefrac{3\mu}{2} + \nicefrac{q}{2}$ if the underlying graph is bipartite, and
\item $y^*(H)  \leq  \nicefrac{5\mu}{3} + \nicefrac{q}{3}$ otherwise.
\end{enumerate}
\end{theorem}

We will prove the claim by induction on $\mu$, the cardinality of the matching in $H$. For that, we will find useful a translation of the following result from \cite{DBLP:conf/mfcs/Stamoulis14} which says that any basic feasible solution for $\mathcal{M}_c$ (and, consequently, the natural linear programming relaxation for the 3-uniform hypergraph matching interpretation of the Rainbow/BCM problem captured by $\mathcal{HM}_c$) has a very particular structure. The result holds on both general and bipartite graphs. We restate the result in terms of hypergraphs as opposed to the pure BCM setting that was originally stated, but the restatement is straightforward. In the following we remind that a \textit{basic feasible solution} (or vertex solution) for an LP is a solution that cannot be written as a convex combination of other feasible solutions.

\begin{theorem}[Lemma 2 in \cite{DBLP:conf/mfcs/Stamoulis14}]\label{thm_degree}
Let $\mathbf{x} \in (0,1)^E$ be a basic feasible solution for the linear relaxation of the 3-uniform hypergraph matching problem. Construct the graph $L_H$ by including a hyperedge $e$ in $E(L_H)$ if $x_e > 0$. Then, there exists some vertex $v \in V(L_H) \subseteq V(H) (= V \cup N)$ such that the degree of $v$ in $L_H$ is at most 2.
\end{theorem}

In other words, basic feasible solutions are \emph{sparse}. Indeed, we can form a basic feasible solution by selecting $|E| = |E(L_H)|$ linearly independent constraints from our linear program, setting them to equality, and solving the linear system. The above result simply says that the number of non-zero variables (corresponding to edges in $L_H$) is equal to the number of linearly independent constraints set to equality. The assumption that $x_e \in (0,1)$ implies that all constraints that we set to equality are vertex constraints, but not non-negativity constraints. See \cite{schrijver, iterative_book} for more details.  We will critically exploit this fact in the following.

\medskip
\noindent{\textbf{Some Notation:}} Before we move on to the proof, we set up some notation. Let $v$ be any vertex of $H$ (in fact, of $L_H$). Denote by $E(v)$ the set of edges that contain $v$ i.e., $E(v) = \{ e \in E(L_H): v \in e \}$. Also, denote by $H(e)$ the hypergraph that is obtained by removing edge $e$ and all edges $e'$ that intersect with $e$ i.e., the hypergraph with  edge set $H(e) = \{ e' \in E(H) \mbox{ such that } e \cap e' = \emptyset \}$. According to Theorem \ref{thm_degree}, we can always find a vertex $v$ of degree at most 2 in $L_H$. Let $e_1$ and $e_2$ be these two edges, with non-zero fractional value $x_{e_1}, x_{e_2}$ respectively, incident on $v$ in $L_H$ and let $H(e_i)$, $i = 1,2$, be the sub-hypergraph obtained by removing $e_i$  and all edges intersecting with this edge. For any $u \in V(L_H) \setminus \{v\}$ let $\delta_v(u) \in \{0,1,2 \}$ be the degree of $u$ in $H_v = (V(H), E(v))$ i.e., the degree of $u$ in the subgraph consisting of the two edges in $E(v)$.

\begin{proof}[Proof of Theorem \ref{main_thm}] With the above notation and relevant results, we will prove the claim of the theorem by an inductive argument on the cardinality of $\mu$.

\medskip
\noindent{\textbf{Base Case:}} For the base case of the induction, assume that $\mu = 1$ (and, of course, $q$ can be at most 1).  It is immediate to see that in this case $H(e_1), H(e_2)$ are both the empty graphs: if not, then we can always choose two independent edges for a matching size $\mu =2$, one edge from $H(e_i)$ and then one edge among $e_1, e_2$ and all other edges that intersect them.  We will construct a feasible solution for the dual LP $D(\mathcal{HM}_c)$ as follows: for every $u \in V(L(H))\setminus \{v\}$ put $y_u = \nicefrac{\delta_v(u)}{2}$. We first claim that this is a feasible solution i.e., satisfies all constraints $\sum_{u \in e} y_u \geq 1$ for all hyperedges $e$. For the base case we need to prove the claim only for $e \in E(v)$ and any other hyperedge that intersect either $e_1$ or $e_2$  (since $H(e_i), i=1,2$ is empty in this case). It is easy to see that for any such edge
$$\sum_{u \in e } y_u =  \frac{1}{2} \sum_{u \in (e\setminus \{v\})} \delta_v(u) \geq \frac{1}{2} \cdot 2 =1.$$

We have used the fact that for each of the two edges in $E(v)$, the remaining vertices in each edge have degree in $H_v$ at least 1. For edges $e \notin E(v)$ we use the fact that, in case $\mu=1$, such edges intersect both $e_1, e_2$  incident on vertex $v$: if that was not the case, then there would be vertices of degree 1 (since $H(e_i) = \emptyset$, for $i=1,2$), a contradiction. In other words, if there was an edge $e \notin E(v)$ that intersects exactly one of $e_1, e_2$,  say intersects only $e_1$, then $e_2 \cup e$ gives a matching of size 2 since $e_2 \cap e = \emptyset$ by assumption which gives a contradiction that there exists an edge that intersects exactly one of the edges in $H(v)$.  Then, we see immediately that the above constraint is satisfied in this case as well.

We will now compute the value of the dual LP which, by duality, will give an upper bound on the fractional solution for the hypergraph matching problem. We have that

\begin{displaymath}
y^*(H) \leq \sum_{u \in V(H)} y_u = \frac{1}{2} \sum_{v \in V(H) \setminus \{v\}} \delta_v(u) \leq \frac{1}{2} \cdot 4 = 2.
\end{displaymath}

We will now show that for the base case ($\mu = 1$), if $q = 0$ then $y^*(H) \leq \nicefrac{3}{2}$ for bipartite and $y^*(H) \leq \nicefrac{5}{3}$ for non-bipartite graphs which will complete the proof for the base case. The fact that $\mu = 1$ means that all edges are pairwise intersecting either on a common ``color" vertex $c_j$ or on a common vertex of the normal graph $G$ (or both). This means that either all edges have the same color, or $G$ is the star graph, or $G$ is a (possibly heterochromatic) triangle (pairwise vertex intersection). In all cases, it is immediate by a simple search to see that the maximum possible fractional value we can get is $\nicefrac{5}{3}$ for general graphs and $\nicefrac{3}{2}$ for bipartite graphs, see Figure \ref{five_thirds}.

\begin{figure}[h!]
\label{five_thirds}
\centering
\includegraphics[scale=0.80]{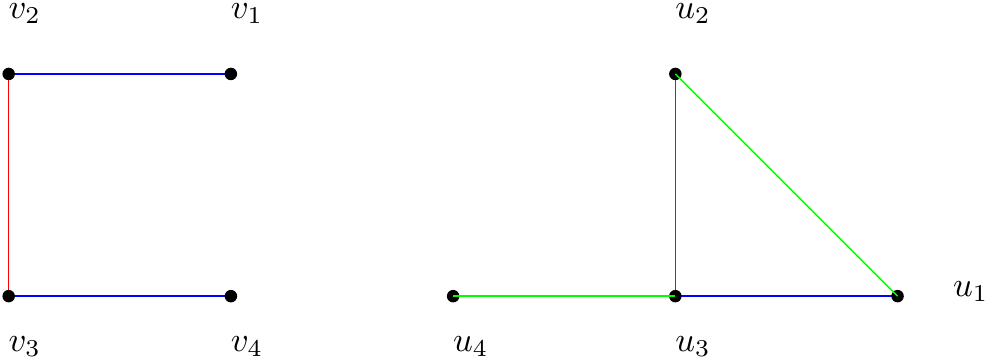}
\caption{Two graphs with matching size $\mu = 1$ and with $q=0$ that achieve the bounds of the base case. On the left, a bipartite graph where edges $\{v_1,v_2\}, \{v_3, v_4\}$ are blue and the other edge is red. By assigning value $\nicefrac{1}{2}$ to each edge we get optimal fractional value of $\nicefrac{3}{2}$ vs. value 1 in the integral case. On the right side we have a graph which is not bipartite: edges $\{u_1, u_2\},  \{ u_3,u_4\}$ are green, edge $\{u_1, u_3\}$ is blue and the remaining edge $\{u_2,u_3\}$ is red. By assigning value $\nicefrac{2}{3}$ to $\{u_1,u_2\}$ and value $\nicefrac{1}{3}$ to each other edge, we get an optimal feasible fractional solution of value $\nicefrac{5}{3}$ vs. value 1 in the integral case.}
\label{five_thirds}
%value $\nicefrac{2}{3}$ to the red edge and $\nicefrac{1}{3}$ to each of the blue edges we get a feasible fractional solution of value \nicefrac{4}{3} }
\end{figure}

\noindent \textbf{Inductive Step:} Now, assume that $G$ is an instance for the Rainbow Matching problem that has a matching of size $\mu$ and has $q$ disjoint copies of $BC_H$. Let $\mathbf{x} \in [0,1]^E$ be a basic feasible fractional solution of $\mathcal{M}_c$ for this instance and let $L_H$ be the restriction of the hypergraph representation of $G$ with respect to $\mathbf{x}$. We know from Theorem \ref{thm_degree} that $L_H$ must have a vertex $v$ with degree (at most) 2 and let $e_1, e_2$ be these two edges incident to $v$. Consider the two sub-hypergraphs $H(e_1)$ and $H(e_2)$ where $H(e_i), i=1,2$ is constructed by removing $e_i$ and all edges intersecting with $e_i$ from $L_H$. These two sub-hypergraphs induce two new (fractional) sub-instances for the underlying restricted matching problem.

In order to use the inductive hypothesis, we need to prove the following easy but important claim:

\begin{claim}
The restrictions $\mathbf{x}_i$ of $\mathbf{x}$ to the edges induced by $H(e_i), i \in \{1,2\}$ are still basic feasible solutions for the corresponding restricted instances of $\mathcal{HM}_c$.
\end{claim}

\begin{proof}
This proof follows an easy pattern but we keep it here for the sake of completeness. Indeed, if they were not, then $x$ would not have been a basic feasible solution for $L_H$ in the first place, since  then it could be written as a convex combination by the corresponding convex decompositions of the restricted vectors of $\mathbf{x}$ in a straightforward way: Assume that for $\alpha \in (0,1)$, $\mathbf{x}_i = \alpha \mathbf{z}_1 + (1-\alpha) \mathbf{z}_2$. Then $\mathbf{x} = \alpha \mathbf{z}_1 + (1-\alpha) \mathbf{z}_2 + \mathbf{x}_e$ where $\mathbf{x}_e \in (0,1)^E$ is simply the restriction of $\mathbf{x}$ to $e_1$ and all edges intersecting it, and has zero on indices of all the other edges (that are in $H(e_1))$. A contradiction.
\end{proof}

This means that we can apply the inductive hypothesis on $H(e_1)$ and $H(e_2)$ separately: each one of them has a matching of size at most $\mu-1$ (since, in each, we have removed an edge and all its intersecting edges which means that we can always add back at least one extra edge to reach $\mu$, the size of matching in $H$) and each one has at most $q$ disjoint copies of $BC_H$. In the following, we will prove the claim only for the case of general graphs but completely identical arguments hold for the bipartite case as well. In our analysis below we will distinguish between two cases: (1) both $H(e_i)$ have a matching of size exactly $\mu-1$, or (2) at least one of $H(e_1)$ has a matching of size $\leq \mu-2$. The first case simply means that there is no edge $e$ that intersects exactly one of $e_1, e_2$. So, in this case, for any edge $e \notin \{e_1, e_2\}$ we have that either $e$ does not intersect any of them or intersects both of them.

We start with the second case. We will define a fractional dual solution $y$ for $L_H$ as follows: $y(v) = 0$ and for any $u\neq v \in V(H)$
\begin{eqnarray}\label{eqn}
y(u) = \frac{1}{2} \Big( \delta_v(u) + y_1(u) + y_2(u) \Big)
\end{eqnarray}

\noindent where the existence of each $\mathbf{y}_i$ is guaranteed by inductive hypothesis.
\begin{claim}
The solution $y$ defined above is a feasible dual solution for the Rainbow Matching problem instance.
\end{claim}

\begin{proof}
To show that this is indeed a feasible dual solution for our instance, we need to show that $\sum_{u \in e} y(u) \geq 1$ for all edges $e$ of $L_H$ (since the non-negativity constraint is trivially satisfied).  This is indeed true for any edge $e \in H_v$ since,  as before in the base case, for any such edge we have that  $\sum_{u \in e} y(u) =   \nicefrac{1}{2} \sum_{u \in e} \delta_v(u) \geq \nicefrac{1}{2} \cdot 2 = 1$.  For all other edges $e \in E(L_H) \setminus E(H_v)$ we have that
\begin{displaymath}
\sum_{u \in e} y(u) =  \frac{1}{2} \sum_{u \in e} \Big( \delta_v(u) + y_1(u) + y_2(u) \Big)
\end{displaymath}
\noindent If edge $e$ intersects both $e_1$ and $e_2$ (the edges incident to the degree-2 vertex $v$ in $V(L_H)$ the existence of which is guaranteed by the properties of basic feasible solutions) then $\sum_{u \in e} \delta_v(u) \geq 2$ and the above expression is $\geq 1$ as required. If $\sum_{u \in e} \delta_v(u) = 0$ this means that $e$  intersects neither $e_1$ nor $e_2$ and thus it belongs to both sub-hypergraphs $H(e_1)$ and $H(e_2)$. This further means that for this edge the corresponding constraint is satisfied by both $y_1$ and $y_2$ i.e., $\sum_{u \in e} y_i(u) \geq 1$, $i = 1,2$ and thus the whole expression above is again $\geq 1$. The only remaining case is when $e$ intersects only one of $e_1, e_2$ in one vertex i.e., the case where $\sum_{u \in e} \delta_v(u) = 1$. Without any loss let us assume that $e$ intersects $e_1$ only. This means that $e \notin E(H(e_1))$ but since $e$ does not intersect $e_2$ we have that $e \in E(H(e_2))$ and as such the constraint is satisfied by the dual solution $y_2$ i.e., $\sum_{u\in e} y_2(u) \geq 1$. This, together with the fact that $\sum_{u \in e} \delta_v(u) = 1$ proves that $y(u) \geq 1$ as required.
\end{proof}

To finish the proof, we will give a bound on the fractional dual solution value which, by duality theory, gives an upper bound on the fractional value (and hence integrality gap) for the fractional Rainbow Matching problem. We remind that we are in the case where at least one of the $H(e_i)$ has a matching of size at most $\mu-2$, let this be $H(e_1)$. We apply the inductive hypothesis on $\mathbf{y}_i$ and  we have that
\begin{eqnarray*}
y^* & = & \sum_{v \in V(L_H)} y(u) = \frac{1}{2} \Bigg( \sum_{u \in V(L_H)\setminus \{v\}} \delta_v(u) + y_1(u) + y_2(u) \Bigg) \\
& = & \frac{1}{2} \Bigg( \sum_{u \in V(L_H)\setminus \{v\}} \delta_v(u) + \sum_{u \in V(L_H)\setminus \{v\}} y_1(u) +  \sum_{u \in V(L_H)\setminus \{v\}} y_2(u) \Bigg) \\
& \leq & \frac{1}{2} \Big( 4 + \frac{5(\mu-2)}{3} + \frac{q}{3} + \frac{5(\mu-1)}{3} + \frac{q}{3} \Big) \\
& = & \frac{5\mu}{3} + \frac{q}{3} - \frac{1}{2} \\
& < & \frac{5\mu}{3} + \frac{q}{3}
\end{eqnarray*}
\noindent as desired.

We now move to the first case: both $H(e_i)$ have a matching of size exactly $\mu-1$ which implies that there is no edge $e$ that intersects exactly one of $e_1, e_2$ at one endpoint (say $e_1$): if there was such an edge then the graph that is induced by edge $e_1$ and all other edges that intersect it would have matching size of $2$ since $e$ and $e_2$ are independent which implies that the matching size of the original instance is $\mu+1+2 > \mu$. Let $R(e_i)$ be the set of edges that intersect $e_i$ (including $e_i$). We know by assumption that each $R(e_i)$ has a matching of size exactly 1. Again, we will distinguish between two cases: (1.a) Neither of $R(e_i)$ is isomorphic to a $BC_H$, and (1.b) both of them are (since there are no edges intersecting exactly 1 of $e_1, e_2$ implies that either both are isomorphic to $BC_H$ or none is).

We start with case (1.a) and we use the inductive hypothesis (in fact the base case since the matching size is one) on $R(e_i)$. This case tells us that there exists a dual solution $\mathbf{y}_{R_i}$ for the  vertices in $R(e_i)$ with value at most $\nicefrac{5}{3}$ and by inductive hypothesis there exist dual solutions $\mathbf{y}_i$ with the desired properties. Define a new dual solution vector $\mathbf{y}$ for $H$ as follows:
$$y(u) = \frac{1}{2} \Big( y_{R_1}(u) + y_{R_2}(u) + y_1(u) + y_2(u) \Big).$$

\begin{claim}
The solution $y$ defined above constitutes a feasible dual solution for the Rainbow Matching problem instance.
\end{claim}

\begin{proof}
As before, all edges in $R(e_i)$ are satisfied by the dual solution by construction. This includes the edges that intersect both $e_1$ and $e_2$. The only remaining edges are those that intersect neither $e_1$ nor $e_2$ and identical arguments to the previous case (2) can be applied here as well.
\end{proof}

We finish this case by providing a bound on the fractional dual solution value where again we apply the inductive hypothesis on $\mathbf{y}_i$ and  we have that
\begin{eqnarray*}
y^* & = & \sum_{v \in V(L_H)} y(u) =   \frac{1}{2} \Bigg( \sum_{u \in V(L_H)}  y_{R_1}(u) + y_{R_2}(u) + y_1(u) + y_2(u) \Bigg) \\
%& = &  \frac{1}{2} \Big( \sum_{u \in V(L_H)}  y_1(u) +  \sum_{u \in V(L_H)} y_2(u) \Big) \\
& \leq & \frac{1}{2} \Big( \frac{5}{3} + \frac{5}{3}+  \frac{5(\mu-1)}{3} + \frac{q}{3} + \frac{5(\mu-1)}{3} + \frac{q}{3} \Big) \\
& = & \frac{5\mu}{3} + \frac{q}{3}
\end{eqnarray*}

\noindent and this concludes the proof of case (1.a). For case (1.b) where both $R(e_i)$ are isomorphic to $BC_H$  we use the solution vector defined in Equation (\ref{eqn}) above. With identical arguments we see that this is a feasible solution vector whose value (using, once more, the inductive hypothesis and the fact that both $H(e_i)$ can have at most $q$ disjoint copies of $BC_H$) is
\begin{eqnarray*}
y^* & = & \sum_{v \in V(L_H)} y(u) = \frac{1}{2} \Bigg( \sum_{u \in V(L_H)\setminus \{v\}} \delta_v(u) + y_1(u) + y_2(u) \Bigg) \\
& = & \frac{1}{2} \Bigg( \sum_{u \in V(L_H)\setminus \{v\}} \delta_v(u) + \sum_{u \in V(L_H)\setminus \{v\}} y_1(u) +  \sum_{u \in V(L_H)\setminus \{v\}} y_2(u) \Bigg) \\
%& = &  \frac{1}{2} \Big( \sum_{u \in V(L_H)}  y_1(u) +  \sum_{u \in V(L_H)} y_2(u) \Big) \\
& \leq & \frac{1}{2} \Big( 4 +  \frac{5(\mu-1)}{3} + \frac{q-1}{3} + \frac{5(\mu-1)}{3} + \frac{q-1}{3} \Big) \\
& = & \frac{5\mu}{3} + \frac{q}{3} + 2 - \frac{5}{3} - \frac{1}{3} \\
& = & \frac{5\mu}{3} + \frac{q}{3}.
\end{eqnarray*}

The case where the underlying graph is bipartite can be handled by completely identical arguments, changing only the bounds.

We finish the proof by briefly comment the case where the vertex $v$ guaranteed by Theorem \ref{thm_degree} has degree 1 (the case zero is of no importance). Let $e_1$ be the unique edge guaranteed by Theorem \ref{thm_degree} incident to vertex $v$ and let $H(e_1)$ be the sub-hypergraph resulting from the removal of $e_1$ and all other edges that intersect it. As before, we apply the inductive hypothesis on $H(e_1)$ which has a matching of size $\mu-1$ and, of course, at most $q$ disjoint copies of $BC_H$. This means that we are in a sub-case of case (1) above where we dealt with the case that both $H(e_1), H(e_2)$ each have a matching of size  exactly $\mu-1$. In this case, $H(e_2)$ is simply the empty graph. We use the adjusted equation \ref{eqn} as above where now $\mathbf{y}_2$ is simply the null (zero) vector and the $\nicefrac{1}{2}$ factor is no longer needed. In other words, we define for any $u \neq v$ $y(u) = d_v(u)+y_1(u)$. Notice that in this case $d_v(u) \in \{0,1\}$. The proof that the proposed vector is indeed a feasible vector follows from the proof of the Claim preceding equation \ref{eqn}. If edge $e$ intersects $e_1$, then indeed we have that $\sum_{u \in e} \delta_v(u) \geq 1$ since in this case $\delta_v(u) = 1$. If $e$ does not intersect $e_1$ then $\sum_{u\in e}y_1(u) \geq 1$ by induction, and we are done. The bound on the fractional dual solution value provided by $y$ is identical to previous cases.
\end{proof}

The above upper bound suggests that the fewer pairwise disjoint truncated projective planes we have in our input graph (more precisely: in its hypergraph representation) the closer to $\nicefrac{3}{2}$ the integrality gap gets and the more we have, the closer to 2 we get - which we know is an upper bound on the integrality gap, achievable by \cite{DBLP:conf/mfcs/Stamoulis14, DBLP:conf/waoa/ParekhP14}.

\subsection{Algorithmic Implications}
One very natural and immediate question is if, and how, the above arguments can be exploited and turned into an explicit algorithmic construction achieving the corresponding bounds.  The natural approach would be to study the cases of BCM where instances are constrained to have only few of these bi-chromatic cycles. Another approach is to add the following linear constraint in the natural LP relaxation of the problem:

\begin{displaymath}
\sum_{e \in BC} x_e \leq 1, ~~~\forall \textrm{ bi-chromatic $4$-cycle } BC.
\end{displaymath}

We call these constraints as ``bi-chromatic constraints". By doing this, we explicitly require that the fractional value assigned to each such bi-chromatic cycle is reduced to 1 (instead of 2, which is the source of the bad integrality gap).  Theorem \ref{main_thm} suggests that the integrality gap of the new enhanced LP would be bounded by $\nicefrac{3}{2}$ in bipartite graphs and $\nicefrac{5}{3}$ in general graphs. However, it is not clear what a rounding procedure would be that can output an integral solution achieving these bounds. The ``vanilla" rounding approach would give again an $\nicefrac{1}{2}$-approximation guarantee.

We note that a very similar idea was pursued in \cite{DBLP:journals/mp/ChanL12} where the authors defined an analogous LP for the 3-Hypergraph Matching problem (maximum matching in 3 uniform hypergraphs). For this problem, the standard LP-based approach gives an approximation guarantee of $\nicefrac{3}{7}$ and, as previously mentioned, F\"{u}redi \cite{DBLP:journals/combinatorica/Furedi81} proved that the so-called \textit{Fano} plane is the only structure that forces the LP to achieve this bound: the Fano plane, the 3-uniform projective plane, achieves integrality gap of exactly $\nicefrac{7}{3}$. Given this, Chan Li \& Lau considered the \textit{Fano} LP: for every Fano plane, add the constraint that the sum of the values of the variables corresponding to edges of the Fano plane is at most 1 (they actually used, for technical reasons, the weaker version that this sum should be at most 2). They proved that the new Fano-LP should have improved integrality gap from $\nicefrac{7}{3}$ to 2. Unfortunately, they did not provide arguments that could make this algorithmic.

We believe that such an algorithmic question requires further insights on the structure of the enhanced LP beyond the ``sparsity" of basic feasible solutions - a property which is not clearly extended to the enhanced LP. Such insights would lead to an appropriate rounding procedure that eludes us at the moment. One promising road could be by noticing that the results of the previous section have also implications on the performance of the Sherali-Adams hierarchy on the natural LP relaxation of the BCM problem. The reason is that these bi-chromatic cycles have very simple structure and we expect that the Sherali-Adams hierarchy might be able to ``recognize" them after few rounds. After all, $r$ rounds of the Sherali-Adams (and other related lift-and-project techniques) can generate all valid ``local" constraints on $r$ variables, so we would expect that 4 rounds should be enough to imply the above set of constraints ($\sum_{e \in BC} x_e \leq 1$, $\forall$ bi-chromatic $4$-cycle BC.) Here we will show that this claim is actually true. For completeness, we will include the details of a slightly strengthened version of this result.

\begin{lemma}\label{lemma:bi-chromatic}
All the bi-chromatic constraints, for every bi-chromatic 4-cycle BC, are implied after 2 rounds of the Sherali-Adams hierarchy applied to the natural LP relaxation of the BCM problem.
\end{lemma}

\begin{proof}
Take a bi-chromatic 4-cycle with edges $e_1,e_2,e_3,e_4$ such that, w.l.o.g., $e_1, e_2$ are painted blue and $e_3, e_4$ red. Thus we have the following two color constraints: $x_{e_1}+ x_{e_2} \leq 1$ and $x_{e_3}+x_{e_4} \leq 1$. At the second round, the Sherali-Adams hierarchy will multiply both the right and the left side of these constraints (according to rule SA-1) with products of two variables. Among (many) others, we have the following constraint
$$ (x_{e_1} + x_{e_2})(1-x_{e_3})(1-x_{e_4}) \leq  (1-x_{e_3})(1-x_{e_4}),$$

\noindent which is equivalent to
$$x_{e_1} + x_{e_2}+x_{e_3}+x_{e_4} \leq 1+x_{e_1}x_{e_3} + x_{e_1}x_{e_4} + x_{e_2}x_{e_3} + x_{e_2}x_{e_4} + x_{e_3}x_{e_4} - x_{e_1}x_{e_3}x_{e_4} - x_{e_2}x_{e_3}x_{e_4} ~~~(A).$$

The product of variables $x_i x_j$ will be simulated/substituted by $y_{ij}$ (rule SA-3) and is, by definition, $\geq 0$ (rule SA-2). Another set of constraints that will be added is the following:
$$(x_{e_1}+x_{e_2})x_{e_1} \leq x_{e_1} ,$$

\noindent from which, using the rule $x_i^2 = x_i$, we get that $x_{e_1}x_{e_2} \leq 0$ which means that $x_{e_1}x_{e_2} = 0$ since the variable is non-negative. Similarly, by multiplying the constraint $x_{e_3}+x_{e_4} \leq 1$ with $x_{e_3}$ we can deduce that $x_{e_3}x_{e_4} = 0$. Now, we work with the degree constraints. For each vertex of this bi-chromatic 4-cycle we have one degree constraint. Take the vertex where edges $e_1, e_3$ meet, and let us name it $v_1$. We have that $x_{e_1}+x_{e_3}+\alpha \leq 1$, where $\alpha$ is simply a term (sum of variables) corresponding the rest of the edges incident to $v_1$. Multiplying this constraint with $x_{e_3}$ we get that $x_{e_1}x_{e_3} + \alpha_{e_3}x_{e_3} \leq 0 \Rightarrow x_{e_1}x_{e_3} = 0$. Identically, we get that $x_{e_1}x_{e_4} = x_{e_2}x_{e_3} = x_{e_2}x_{e_4} = 0$.  This means that inequality (A) from above becomes
$$x_{e_1} + x_{e_2}+x_{e_3}+x_{e_4} \leq 1 - x_{e_1}x_{e_3}x_{e_4} - x_{e_2}x_{e_3}x_{e_4} \leq 1, $$
i.e., 2 rounds of the Sherali-Adams hierarchy imply all the bi-chromatic 4-cycle constraints.
\end{proof}

Combining Lemma \ref{lemma:bi-chromatic} with Theorem \ref{main_thm} gives the following:

\begin{theorem}
Let $G$ be a graph with $q \in \mathbb{Z}^{\geq 0}$ disjoint copies of bi-chromatic 4-cycles. After 2 rounds of the Sherali-Adams hierarchy applied to the natural LP relaxation of the BCM problem, the integrality gap is at most $\nicefrac{5}{3}$ for general graphs and at most $\nicefrac{3}{2}$ for bipartite graphs.
\end{theorem}

This means that instead of working with the enhanced LP, we can work directly with hierarchies generated by a low number of rounds, and use promising rounding approaches.  We leave this as an open question with the hope that the results of the current manuscript will be a first step towards this direction. Given the similarity of BCM with the 3-hypergraph matching problem, we also expect that any positive result for the former could be adapted to provide a positive result for the later, solving a major open problem in the field.

\section{Acknowledgements}
The authors would like to sincerely thank the two anonymous referees and the editorial for carefully reading a first version of the manuscript and for the many helpful comments and suggestions that greatly improved its content and presentation.  The second author acknowledges the support of an NWO TOP 2 grant (617.001.301).

\bibliographystyle{abbrv}
\bibliography{references}

\end{document}